\documentclass[12pt,oneside]{amsart}
     \usepackage{amssymb}
      \usepackage[utf8]{inputenc}
     \usepackage[frenchb,english]{babel}
      \usepackage[T1]{fontenc}
      \usepackage{amsmath}
      \usepackage{amsfonts}
      \usepackage{lmodern}
      \usepackage[left=2cm,right=2cm,top=2cm,bottom=2cm]{geometry}
      \usepackage{variations}
      \usepackage{graphicx} 
      \usepackage{float}	
      \usepackage{geometry} 
      \usepackage{fancyhdr} 
      \usepackage{longtable}
      \usepackage{listings}
      \usepackage{subfigure}
      \usepackage{hyperref} 
      \usepackage[usenames,dvipsnames]{color} 

      \theoremstyle{plain}
      \newtheorem{theorem}{Theorem}[section]
      \newtheorem{lemma}[theorem]{Lemma}
      \newtheorem{corollary}[theorem]{Corollary}
      
      \newtheorem{Proposition}[theorem]{Proposition}
      
      \theoremstyle{definition}
      \newtheorem{definition}[theorem]{Definition}
      
      \theoremstyle{remark}
      \newtheorem{remark}[theorem]{Remark}

      \numberwithin{equation}{section}

      \makeatletter
      \def\@setcopyright{}
      \def\serieslogo@{}
      \makeatother

\begin{document}

   \author{V. Maume-Deschamps}
   \address{Université de Lyon, Université Lyon 1,France, Institut Camille Jordan UMR 5208}
   \email{veronique.maume@univ-lyon1.fr}


   \author{D. Rulli\`ere}

   \address{Université de Lyon, Université Lyon 1,France\\
    Laboratoire SAF EA 2429}

   \email{didier.rulliere@univ-lyon1.fr}
   \author{K. Said}
      \address{Université de Lyon, Université Lyon 2,France, Laboratoires SAF EA 2429 \& COACTIS EA 4161 }
      \email{khalil.said@univ-lyon2.fr}
   

   \title[The Infinite]{A risk management approach to capital allocation}

   \begin{abstract}
   The European insurance sector will soon be faced with the application of Solvency 2 regulation norms. It will create a real change in risk management practices. The ORSA approach of the second pillar makes the capital allocation an important exercise for all insurers and specially for groups. Considering multi-branches firms, capital allocation has to be based on a multivariate risk modeling. Several allocation methods are present in the literature and insurers practices. In this paper, we present a new risk allocation method, we study its coherence using an axiomatic approach, and we try to define what the best allocation choice for an insurance group is.
   \end{abstract}										
   \date{\today}
   \maketitle
   \section*{Introduction}
       Solvency 2 standards will make a radical change in risk management practices in the actuarial sector. They are based on a strengthening of risk control and minimization of a ruin probability. The determination of the economic regulatory capital will be faced under this prudential mechanism with a kind of methodological revolution. A choices of a dependence model between different risks and of an aggregation methodology are both required. In the prescribed \textit{Standard Formula}, risks aggregation is done using correlation matrices that connect families and subfamilies of risks. Once the \textit{Solvency Capital Requirement (SCR)} is calculated, its allocation between the different risky activity branches becomes the new operational challenge.\\
       
        Capital allocation is an internal exercise, certainly not controlled by the first pillar of Solvency~2, but it plays a crucial role in determining performance of all the insurer activity. The case of insurance groups requires special treatment in the context of the \textit{Own Risk and Solvency Assessment (ORSA)} approach. In this context, a multivariate analysis of risk seems relevant. \\
        
        The issue of capital allocation in a multivariate context arises from the presence of dependence between the various risky activities which may generate a diversification effect. Several allocation methods in the literature are based on a choice of a univariate risk measure and an allocation principle. Others are based on optimizing a multivariate ruin probability or some multivariate risk indicators. In this paper, we focus on the allocation technique by minimizing some risk indicators.\\
        
     The literature on the subject of capital allocation methods is very rich. Several principles have been proposed over the last twenty years. The most important and most studied are the Shapley method, the Aumann-Shapley method and the Euler's method.\\ The Shapley method is based on cooperative game theory. It is described in detail in Denault's paper (2001) \cite{Denault2001}. Denault proved that this method, originally used to allocate the total cost between players in coalitional games context, can be easily adapted to solve the problem of the overall risk allocation between segments.\\
   Tasche devoted two papers \cite{Euler2} and \cite{Euler} to describe Euler's method. Euler's method is also found in the literature under the name of gradient method. It is based on the idea of allocating capital according to the infinitesimal marginal impact of each risk. This impact corresponds to the increase obtained on the overall risk, yielding an infinitely small increment in a marginal risk. Euler's method is very present in the literature. Several papers analyze its properties (RORAC compatibility \cite{RORAC}, \cite{RAROC1}, Coherence \cite{Alloc2},...) under different assumptions (Tasche (2004) \cite{Euler3}, Balog (2011) \cite{Euler4}). Its fame is due to the existence of economic arguments that can justify its use to develop allocation rules.\\
   Finally, Aumann-Shapley method is a continuous generalization of Shapley method. Its principle is based on the value introduced by Aumann and Shaplay in game theory. Denault \cite{Denault2001} analyzes this method and its application to capital allocation.\\
   
   These three capital allocation principles rely on different risk measures. The coherence of the allocation method depends on the properties of the selected risk measure. Several papers deal with capital allocation coherence based on the properties of the risk measure used. We quote as examples, Fischer (2003) \cite{Fischer2003}, Bush and Dorfleitner (2008) \cite{Alloc2}, and Kalkbrener (2009) \cite{coherence2}.\\
   
   Other techniques have been proposed more recently for building optimal allocation methods, by minimizing some multivariate ruin probabilities, especially those defined by Cai and Li (2007) \cite{Cai}, or by minimizing some new multivariate risk indicators. In this context, Cénac et al. \cite{AR1}, \cite{AR2} defined three types of indicators, which take into account both the ruin severity at the branch level, and the impact of the dependence structure on this local severity. In the one-period case, these indicators can be considered as special cases of a general indicator family introduced in Dhaene et al. (2012) \cite{AR3}. Allocation by minimizing these indicators was studied in bivariate dimension by Cénac et al. (2014) \cite{AR1}, and in higher dimension by Maume-Deshamps et al. (2015) \cite{P22015}. In \cite{P22015} we study its behavior and its asymptotic behavior for some special distributions families. We study also in the same paper, the impact of dependence on the allocation composition. In the present paper we focus on the coherence properties of this kind of allocation methods.\\

    Allocation by minimization of some real multivariate risk indicators can be used in a more general framework for modeling systemic risk. In reinsurance, it can also be used to find optimal stop-loss treaties in some special cases. This allocation technique can also help to measure the performance of calculating the groups' capital requirement in the Swiss Solvency Test (SST), which provides a consistent framework both for legal branches and group solvency capital requirement.\\
     
     The article is organized as follows. In the first section, we present the optimal allocation method by minimizing multivariate risk indicators. Using an axiomatic approach, we define in section 2 some coherence properties for allocation methods in multivariate context. The third section is devoted to the study of the coherence of the optimal allocation. Section 4 is a discussion about the best allocation method choice for an insurance group. 
       
   \section{Optimal allocation presentation}
 
     In a multivariate risk framework, we consider a vectorial risk process $X^p=(X_1^p,\ldots,X_d^p)$, where $X_k^p$ corresponds to the losses  of the $k^{th}$ business line during the $p^{th}$ period.  We denote by  $R_k^p$ the reserve of the $k^{th}$ line at time $p$, so: $R_k^p=u_k-\displaystyle\sum_{l=1}^{p}X_k^l$, where $u_k\in\mathbb{R}^+$ is the initial capital of the $k^{th}$ business line. $u=u_1+\cdots+u_d$ is the initial capital of the group and $d$ is the number of business lines.\\
     
   Cénac et al. (2012) \cite{AR2} defined the two following multivariate risk indicators, for $d$ risks and $n$~periods, given penalty functions $g_k, ~k\in\{1,\ldots,d\}$ :
   \begin{itemize}
   \item the indicator $I$: \[ \mathit{I}(u_1,\ldots,u_d)=\sum_{k=1}^{d}{\mathbb{E}\left(\sum_{p=1}^{n}{g_k(R^p_k)1\!\!1_{\{R^p_k<0\}}1\!\!1_{\{\sum_{j=1}^{d}R^p_j>0\}}}\right)} \/,\]
   \item the indicator $J$: 
    \[ \mathit{J}(u_1,\ldots,u_d)=\sum_{k=1}^{d}{\mathbb{E}\left(\sum_{p=1}^{n}{g_k(R^p_k)1\!\!1_{\{R^p_k<0\}}1\!\!1_{\{\sum_{j=1}^{d}R^p_j<0\}}}\right)} \/,\]
   \end{itemize}
   $g_k: \mathbb{R}^-\rightarrow \mathbb{R}^+$ are $C^1$, convex functions with $g_k(0)=0\/, ~g_k(x)\geq 0$ for $x < 0,~k=1,\ldots,d$. They aggregate the cost that each branch has to pay when it becomes insolvent while the group is solvent for the $I$ indicator, or while the group is also insolvent in the case of the $J$ indicator.\\
   They proposed to allocate some capital $u$ by minimizing these indicators. The idea is to find an allocation vector $(u_1,\ldots,u_d)$ that minimizes the indicator such as $u=u_1+\cdots+u_d$, where $u$ is the initial capital that need to be shared among all branches.\\
   
   The indicator $I$ represents the expected sum of penalty amounts of local ruins, knowing that the group remains solvent. In the case of the indicator $J$, the local ruin severities are taken into account only in the case of group insolvency.\\
   
  By using optimization stochastic algorithms, we may estimate the minimum of these risk indicators. Cénac et al. (2012) \cite{AR2} propose a Kiefer-Wolfowitz version of the mirror algorithm as a convergent algorithm under general assumptions to find an optimal allocation minimizing the indicator $I$. This algorithm is effective to solve the optimal allocation problem, especially, for a large number of business lines, and for allocation over several periods.\\
   
   \subsection{Definitions and notations}
   Since new regulation rules, such as Solvency 2, require only a justified allocation over a period of one year, we focus in this paper on the case of  allocations on a single period ($n=1$). Another goal of this choice is to present a first computational approach. An annual allocation seems to be a more efficient decision for an insurer; during a year of operation, it will allow him to integrate the changes that occurred in his risk portfolio and its dependence structure.\\
   
   The following notations are used:
   \begin{itemize} 
   \item[$\square$] $u$ is the initial capital of the firm.
   \item[$\square$]$\mathcal{U}^d_u=\{v=(v_1, \ldots, v_d)\in[0,u]^d, \sum_{i=1}^{d}v_i=u\}$ is the set of possible allocations of the initial capital $u$. 
   \item[$\square$] For all $i\in \{1,\ldots,d\}$ let $\alpha_i=\frac{u_i}{u}$, then, $\sum_{i=1}^{d}\alpha_i=1$ if $(u_1,\ldots,u_d)\in\mathcal{U}^d_u$.
   \item[$\square$]$1\!\!1_u^d=\{\alpha=(\alpha_1, \ldots, \alpha_d)\in[0,1]^d, \sum_{i=1}^{d}\alpha_i=1\}$ is the set of possible allocation percentages $\alpha_i=u_i/u$. 
   \item[$\square$] The risk $X_k$ corresponds to the losses of the $k^{th}$ branch during one period. It is a positive random variable in our context.
   \item[$\square$] For $(u_1,\ldots,u_d)\in\mathcal{U}^d_u$, we define the reserve of the $k^{th}$ business line at the end of the period is: $R^k=u_k-X_k$, where $u_k$ represents the part of capital allocated to the $k^{th}$ branch.
   \item[$\square$] The aggregate sum of risks is:  $ S=\sum_{i=1}^{d}{X_i} $, and let  $ S^{-i}=\sum_{j=1;j\neq i}^{d}{X_j} $ for all $i\in \{1,\ldots,d\}$.
   \item[$\square$] $F_Z$ is the cumulative distribution function of a random variable $Z$, $\bar{F}_Z$ is its survival function and $f_Z$ its probability density function. 
       \end{itemize}
       \begin{definition}[Optimal allocation]
          Let $\mathbf{X}$ be a positive random vector of $\mathbb{R}^d$, $u \in \mathbb{R}^+$ and $\mathcal{K}_\mathbf{X}:\mathcal{U}^d_u\rightarrow \mathbb{R}^+$ a multivariate risk indicator associated to $\mathbf{X}$ and $u$. An optimal allocation of the capital $u$ for the risk vector $\mathbf{X}$ is defined by:
          \[ (u_1,\ldots,u_d)\in\underset{(v_1,\ldots,v_d)\in\mathcal{U}^d_u}{\arg\inf}\left\{ \mathcal{K}_\mathbf{X}(v_1,\ldots,v_d)\right\}\/. \] 
          \end{definition}
          
          For risk indicators of the form $\mathcal{K}_\mathbf{X}(v)=\mathbb{E}[S(\mathbf{X},\mathbf{v})]$, for a scoring function $S:{\mathbb{R}^{+}}^d\times{\mathbb{R}^{+}}^d\rightarrow\mathbb{R}^{+}$, this definition can be seen as an extension in a multivariate framework of the elicitability concept. Elicitability has been introduced by Gneiting (2011)\cite{Gneiting2011}, and studied recently for univariate risk measures, by Bellini and Bignozzi (2013)\cite{Bellini2013}, Ziegel (2014)\cite{Ziegel2014} and Steinwart et al. (2014)\cite{Cachan2014}, for examples.\\
 
  \paragraph{\textbf{Assumptions}}
   Throughout this paper, we will use the following assumptions:
   \begin{description}
   \item[H1] The risk indicator $\mathcal{K}_\mathbf{X}$ admits a unique minimum in $\mathcal{U}^d_u$. In this case, we denote by $ A_ {X_1,\ldots, X_d} (u) = (u_1,\ldots, u_d) $ the optimal allocation of the amount $u$ on the $d$ risky branches in $\mathcal{U}^d_u$.  
   \item[H2]The functions $g_k$ are differentiable and such that for all $k\in\{1,\ldots,d\}$, $g^\prime_k(u_k-X_k)$ admits a moment of order one, and $(X_k,S)$ has a joint density distribution denoted by $f_{(X_k,S)}$. 
   \item[H3] The $d$ risks have the same penalty function $g_k=g, \forall k\in \{1,\ldots,d\}$.
   \end{description}
   The first assumption is verified when the indicator is strictly convex, this is particularly true when for at least one $k\in\{1,\ldots,d\}$, $g_k$ is strictly convex; and the joint density $f_{(X_k,S)}$ support contains $[0,u]^2$ (see \cite{AR2}).
   \subsection{Optimality conditions}
       In this section, we focus on the optimality condition for the indicators $I$ and $J$.\\
        For an initial capital $u$, and an optimal allocation minimizing the multivariate risk indicator $I$, we seek $u^*\in \mathbb{R}_+^d$ such that:
     \[\mathit{I}(u^*)=\underset{v_1+\cdots+v_d=u}{\inf}\mathit{I}(v), ~~v\in \mathbb{R}_+^d  \/.\]
      Under assumption H2, the risk indicators $I$ and $J$ are differentiable, and in this case, we may calculate the following gradients: 
           \begin{align*}
           (\nabla I(v))_i&=\sum_{k=1}^{d}\int_{v_k}^{+\infty}g_k(v_k-x)f_{X_k,S}(x,u)dx
           + \mathbb{E}[g^\prime_i(v_i-X_i)1\!\!1_{\{X_i>v_i\}}1\!\!1_{\{S\leq u\}}]\\
           \mbox{and,}~~ (\nabla J(v))_i&=\sum_{k=1}^{d}\int_{v_k}^{+\infty}g_k(v_k-x)f_{X_k,S}(x,u)dx
                 + \mathbb{E}[g^\prime_i(v_i-X_i)1\!\!1_{\{X_i>v_i\}}1\!\!1_{\{S\geq u\}}]\/.
           \end{align*} 
     Under H1 and H2, using the Lagrange multipliers method, we obtain an optimality condition verified by the unique solution to this optimization problem:
     \begin{equation}
     \mathbb{E}[g^\prime_i(u_i-X_i)1\!\!1_{\{X_i>u_i\}}1\!\!1_{\{S\leq u\}}]=\mathbb{E}[g^\prime_i(u_j-X_j)1\!\!1_{\{X_j>u_j\}}1\!\!1_{\{S\leq u\}}],~~\forall j\in\{1,\ldots,d\}^2
     \label{Optcg}\/.
     \end{equation}
     A natural choice for penalty functions is the ruin severity: $g_k(x)=|x|$. In that case, and if the joint density $f_{(X_k,S)}$ support contains $[0,u]^2$, for at least one $k\in\{1,\ldots,d\}$, our optimization problem  has a unique solution.\\ 
    We may write the indicators as follows:
   \begin{align*}
   \mathit{I}(u_1,\ldots,u_d)&=\sum_{k=1}^{d}{\mathbb{E}\left({\arrowvert R^k \arrowvert1\!\!1_{\{R^k<0\}}1\!\!1_{\{\sum_{i=1}^{d}R^i\geq0\}}}\right)}\\
   &=\sum_{k=1}^{d}{\mathbb{E}\left({(X_k-u_k)1\!\!1_{\{X_k>u_k\}}1\!\!1_{\{\sum_{i=1}^{d}X_i\leq u\}}}\right)}=\sum_{k=1}^{d}{\mathbb{E}\left({(X_k-u_k)^+1\!\!1_{\{S\leq u\}}}\right)}\/,
   \end{align*}
   and,
   \begin{align*}
   \mathit{J}(u_1,\ldots,u_d)&=\sum_{k=1}^{d}{\mathbb{E}\left({\arrowvert R^k \arrowvert1\!\!1_{\{R^k<0\}}1\!\!1_{\{\sum_{i=1}^{d}R^i\leq0\}}}\right)}\\
   &=\sum_{k=1}^{d}{\mathbb{E}\left({(X_k-u_k)1\!\!1_{\{X_k>u_k\}}1\!\!1_{\{\sum_{i=1}^{d}X_i\geq u\}}}\right)}=\sum_{k=1}^{d}{\mathbb{E}\left({(X_k-u_k)^+1\!\!1_{\{S\geq u\}}}\right)}\/.
   \end{align*} 
   The respective components of the gradient of these indicators are of the form:
  \[K_I-\mathbb{P}\left( X_1>u_{1},\sum_{j=1}^{d}{X_j}\leq u\right),\ldots,K_I-\mathbb{P}\left( X_d>u_{d},\sum_{j=1}^{d}{X_j}\leq u\right) \/,\]
  and, 
  \[K_J-\mathbb{P}\left( X_1>u_{1},\sum_{j=1}^{d}{X_j}\geq u\right),\ldots,K_J-\mathbb{P}\left( X_d>u_{d},\sum_{j=1}^{d}{X_j}\geq u\right) \/,\]
  where,
  \[ K_I=K_J=\sum_{k=1}^{d}\int_{u_k}^{+\infty}(x-u_k)f_{X_k,S}(x,u)dx \/.\]
  Using the Lagrange multipliers to solve our convex optimization problem under the only constraint $u_1+u_2+\cdots+u_d=u$, the following optimality conditions are obtained from \ref{Optcg} in the special case where $g_k(x)=|x|$:
  \begin{equation}\label{ZoneO}
  \mathbb{P}\left( X_i>u_{i}, S\leq u\right) = \mathbb{P}\left( X_j>u_{j}, S\leq u\right) , \forall (i,j)\in\{{{1,2,\ldots,d}}\}^{2}\/.
  \end{equation}
  For the $J$ indicator, this condition can be written:
  \begin{equation}\label{ZoneV}
  \mathbb{P}\left( X_i>u_{i}, S\geq u\right) = \mathbb{P}\left( X_j>u_{j}, S\geq u\right) , \forall (i,j)\in\{{{1,2,\ldots,d}}\}^{2}\/.
  \end{equation}
  Some explicit and semi-explicit formulas for the optimal allocation can be obtained with this optimality condition. Our problem reduces to the study of this allocation depending on the nature of the distributions of the risk $X_k$ and on the form of dependence between them.\\
    \section{Coherence of a capital allocation in a multivariate context}   
      In his article \cite{Denault2001}, Denault introduced the notion of a coherent allocation, fixing four axioms that must be verified by a principle of capital allocation in order to be qualified as coherent. Denault's definition can be used only for allocation methods driven by univariate risk measures, especially coherent ones, according to the criteria defined by Artzner et al. (1999) \cite{Artz}. Our optimal capital allocation is not directly derived from a univariate risk measure, even if it is obtained by minimizing a multivariate risk indicator.\\
      
      In this section, we reformulate coherence axioms in a more general multivariate context. We define also other coherence properties and we try to justify for each one why it is a desirable property from an economic point of view.   
      \subsection{Coherence}
             We follow Denault's idea to define a coherent capital allocation in a multivariate context.
       \begin{definition}[Coherence]
       A capital allocation $(u_1,\ldots, u_d)=A_ {X_1,\ldots, X_d} (u)$ of an initial capital $u\in\mathbb{R}^+$ is coherent if it satisfies the following properties: 
       \begin{enumerate}
          \item[\textbf{1.}]\textbf{Full allocation:} All of the capital $u\in\mathbb{R}^+$ must be allocated between the branches:\[ \sum_{i=1}^{d}u_i=u \/. \]
          \item[\textbf{2.}]\textbf{Symmetry:} If the joint distribution of the vector $(X_1,\ldots,X_d)$ is unchanged by permutation of the risks $X_i$ and $X_j$, then the allocation remains also unchanged by this permutation, and the $i^{th}$ and $j^{th}$ business lines both
          make the same contribution to the risk capital:
          if
          \begin{align*}
          (X_1,\ldots,X_{i-1},X_i,X_{i+1},&\ldots,X_{j-1},X_j,X_{j+1},\ldots,X_d)\\ &\stackrel{\mathcal{L}}{=}\\ (X_1,\ldots,X_{i-1},X_j,X_{i+1},&\ldots,X_{j-1},X_i,X_{j+1},\ldots,X_d)\/,
          \end{align*}
          then $u_i = u_j$.
          \item[\textbf{3.}]\textbf{Riskless allocation:} For a deterministic risk $X=c$, where the constant $c\in{\mathbb{R}^+}$:\[ A_{X,X_1,\ldots,X_d}(u)=(c, A_{X_1,\ldots,X_d}(u-c)) \/.\]
           This property means that the allocation method relates only risky branches, the presence of a deterministic risk has no impact on the share allocated to the risky branches.
          \item[\textbf{4.}]\textbf{Sub-additivity:} $\forall M\subseteq \{1,\ldots,d\}$, let $(u^*,u^*_1,\ldots,u^*_r)=A_ {\sum_{i\in M}X_i,X_{j\in\{1,\ldots,d\}\backslash M}} (u)$, where $r=d-card(M)$ and $(u_1,\ldots,u_d)=A_ {X_1,\ldots, X_d} (u)$:\[ u^*\leq \sum_{i\in M}u_i \/.\]
             This property means that the allocation takes into account the diversification gain. It is related to the \textit{no undercut} property defined by Denault, which has no sense in our context.
             \item[\textbf{5.}]\textbf{Comonotonic additivity:} For $r\leqslant d$ comonotonic risks,
                \[ A_{X_{i_{i\in\{1,\ldots,d\}\setminus CR}},\sum_{k\in CR}^{}X_k}(u)=(u_{i_{i\in\{1,\ldots,d\}\setminus CR}},\sum_{k\in CR}^{}u_k)\/, \]
                where $(u_1,\ldots,u_d)=A_ {X_1,\ldots, X_d} (u)$ is the allocation of $u$ on the $d$ risks $(X_1,\ldots,X_d)$ and $CR$ denotes the set of the $r$ comonotonic risk indexes.\\
                The concept of comonotonic random variables is related to the studies of Hoeffding (1940) \cite{Hoeff1940} and Fréchet (1951) \cite{Frechet1951}. Here we use the definition of comonotonic risks as it was first mentioned in the actuarial literature in Borch (1962) \cite{Borch1962}.\\
                A vector of random variables $(X_1,X_2,\ldots,X_n)$ is comonotonic if and only if there exists a random variable $Y$ and non-decreasing functions $\varphi_1,\ldots,\varphi_n$ such that: 
                \[ (X_1,\ldots,X_n)=_d(\varphi_1(Y),\ldots,\varphi_n(Y)) \/.\]
                \end{enumerate}
            \label{defC}
       \end{definition}
       
       \subsection{Other desirable properties}
       We define also some desirable properties that an allocation should naturally satisfy. These properties are based on the ideas presented by Artzner et al. (1999) \cite{Artz} for coherent risk measures and on the axiomatic characterization of coherent capital allocations given by Kalkbrener (2009) \cite{coherence2}.
       \begin{definition}[Positive homogeneity]
       An allocation is positively homogeneous, if for any $\alpha\in \mathbb{R}^+$, it satisfies:
          \[ A_{\alpha X_1,\ldots,\alpha X_d}(\alpha u)=\alpha A_{X_1,\ldots,X_d}(u) \/.\]
         \end{definition}
         In other words, a capital allocation method is positively homogeneous, if it is insensitive to cash changes.
         \begin{definition}[Translation invariance]
          An allocation is invariant by translation, if for all $(a_1,\ldots,a_d)\in \mathbb{R}^d$, it satisfies:
             \[ A_{X_1-a_1,\ldots,X_d-a_d}(u)=A_{X_1,\ldots,X_d}\left(u+\sum_{k=1}^{d}a_k\right)-(a_1,\ldots,a_d)\/. \]
           \end{definition}
             The translation invariance property shows that the impact of an increase (decrease) of a risk by a constant amount of its share of allocation of the capital $u$, boils down to an increase (decrease) of its share in the allocation of such capital decreased (increased) by the same amount.
              \begin{definition}[Continuity]
                An allocation is continuous, if for all $i \in \{1,\ldots,d\}$:
                 \[ \lim\limits_{\epsilon \to 0}A_{X_1,\ldots,(1+\epsilon)X_i,\ldots,X_d}(u)=A_{X_1,\ldots,X_i,\ldots,X_d}(u) \/.\]
              \end{definition}
              This property reflects the fact  that a small change to the risk of a business line, have only limited effect on the capital part that we attribute to it.\\  
                           
               Let us recall the definition of the order stochastic dominance, as it is presented in Shaked and Shanthikumar (2007)\cite{OrdreSto}. For random variables $X$ and $Y$, $Y$ first-order stochastically dominates $X$ if and only if:
                \[\bar{F}_X(x)\leq\bar{F}_Y(x), ~~\forall x\in\mathbb{R}^+ \/,\]
                and in this case we denote: $X\leq_{st}Y$.\\ 
                This definition is also equivalent to the following one:
                \[X\leq_{st}Y\Leftrightarrow \mathbb{E}[u(X)]\leq \mathbb{E}[u(Y)], \mbox{for all } u~ \mbox{increasing function}  \]
              \begin{definition}[Monotonicity]
                 An allocation satisfies the monotonicity property, if for $(i,j)\in \{1,\ldots,d\}^2$:
                 \[ X_i\leq_{st}X_j \Rightarrow u_i\leq u_j \/.\]
               \end{definition}
                The monotonicity is a natural requirement, it reflects the fact that if a branch $X_j$ is riskier than branch $X_i$. Then, it is natural to allocate more capital to the risk $X_j$.\\
                                    
                  The RORAC compatibility property defined by Dirk Tasche \cite{Euler2} loses its meaning in absence of the risk measure used in the construction of the allocation method.\\


   
  
  \section{Coherence of the optimal allocation}
   
      In what follows, we show that the capital allocation minimizing the indicator $I$, satisfies the coherence axioms of Definition \ref{defC}, except the sub-additivity. We show also that it satisfies other desirable properties in the second subsection. The same holds for the indicator $J$.
   \subsection{Coherence}
      Firstly, the \textit{full allocation} axiom is verified by construction, since any optimal allocation satisfies the equality: \[ \sum_{i=1}^{d}u_i=u \/.\]
   Proposition \ref{symP} shows that the optimal allocation satisfies the symmetry property.
   \begin{Proposition}[\textbf{Symmetry}]
    Under H1, if for $(i,j)\in\{{{1,2,\ldots,d}}\}^{2}$, $i\neq j$, the couples $(X_i,S^{-i})$ and $(X_j,S^{-j})$ are identically distributed and the penalty functions $g_i$ and $g_j$ are the same $g_i=g_j$, then: \[  u_i=u_j \/.\]
    \label{symP}
   \end{Proposition}
   \begin{proof}
    Let $ (i\neq j)\in\{{{1,2,\ldots,d}}\}^{2}$ be such that $(X_i,S^{-i})$ and $(X_j,S^{-j})$ have the same distribution and the same penalty function $g_i=g_j=g$. If $ u_i\neq u_j$, we may assume $i<j$, and denote:
    \[ (u_1,\ldots,u_i,\ldots,u_j,\ldots,u_d)=A_{X_1,\ldots,X_i,\ldots,X_j,\ldots,X_d}(u) \/,\]
    then, 
    \[ I(u_1,\ldots,u_i,\ldots,u_j,\ldots,u_d)=\underset{v \in \mathcal{U}^{d}_u}{\inf}\mathit{I}(v)=\underset{v \in \mathcal{U}^{d}_u}{\inf}\sum_{k=1}^{d}{\mathbb{E}\left({g_k(v_k-X_k)1\!\!1_{\{X_k>v_k\}}1\!\!1_{\{S\leq u\}}}\right)} \/.\]
    On the other hand, and since $g_i=g_j=g$ and $(X_i,S^{-i})\sim (X_j,S^{-j})$, then: \begin{align*}
     I(u_1,\ldots,u_{i-1},u_j,u_{i+1},\ldots,u_{j-1},u_i,u_{j+1},\ldots,u_d)&=\sum_{k=1,k\neq i, k\neq j}^{d}{\mathbb{E}\left({g_k(u_k-X_k)1\!\!1_{\{X_k>u_k\}}1\!\!1_{\{S\leq u\}}}\right)}\\&+\mathbb{E}\left({g(u_i-X_i)1\!\!1_{\{X_i>u_i\}}1\!\!1_{\{S\leq u\}}}\right)\\&+\mathbb{E}\left({g(u_j-X_j)1\!\!1_{\{X_j>u_j\}}1\!\!1_{\{S\leq u\}}}\right)\\
     &=I(u_1,\ldots,u_i,\ldots,u_j,\ldots,u_d)\/.
     \end{align*}
    From H1, the indicator $I$ admits a unique minimum in $\mathcal{U}^{d}_u$, we deduce that:\[(u_1,\ldots,u_i,\ldots,u_j,\ldots,u_d)= (u_1,\ldots,u_{i-1},u_j,u_{i+1},\ldots,u_{j-1},u_i,u_{j+1},\ldots,u_d)\/.\] 
    We conclude that $u_i = u_j$ .
   \end{proof}
   \begin{corollary}
   Under Assumptions H1 and H3, if  $(X_1,\ldots,X_d)$ is an exchangeable random vector, then the allocation by minimizing $I$ and $J$ indicators is the same and is given by: \[  A_{X_1,\ldots,X_d}(u)=\left(\frac{u}{d},\frac{u}{d},\ldots,\frac{u}{d}\right)
    \/.  \]
   \end{corollary}
 The following proposition shows that the optimal allocation verifies the Riskless allocation axiom.
 \begin{Proposition}[\textbf{Riskless Allocation}]
 Under Assumptions H1 and H3, and for 1-homogeneous penalty functions, for any $c\in \mathbb{R}$:
       \[ A_{c,X_1,\ldots,X_d}(u)=(c, A_{X_1,\ldots,X_d}(u-c))\/, \]
 where $(c, A_{X_1,\ldots,X_d}(u-c))$ is the concatenated vector of $c$ and the vector $A_{X_1,\ldots,X_d}(u-c)$. 
 \end{Proposition}
 \begin{proof}
 The presence of a discrete distribution makes the indicator $I$ not differentiable, so we cannot use neither the gradient, nor the optimality condition obtained in the case of existence of joined densities.\\ 
  Let 
             $(u^*,u_1^*,\cdots,u_d^*) = A_{c,X_1,\ldots,X_d}(u)$ and $(u_1,\cdots,u_d)=A_{X_1,\ldots,X_d}(u-c)$.\\
 We denote $S=\sum_{i=1}^{d}X_i$, and the common penalty function $g=g_k, \forall k\in\{1,\ldots,d\}$, the function $g$ is convex on $\mathbb{R}^-$ and $g(0)=0$, we deduce that $g$ is also positively homogeneous.\\
 We distinguish between three possibilities:
 \begin{itemize}
 \item \textit{Case 1: $u^*<c$}\\
 In this case, 
 \begin{align*}
 I(u^*,u^*_1,\ldots,u^*_d)&=\underset{v \in \mathcal{U}^{d+1}_u}{\inf}\mathit{I}(v)=\underset{v \in \mathcal{U}^{d+1}_u}{\inf}\sum_{k=0}^{d}{\mathbb{E}\left({g(v_k-X_k)1\!\!1_{\{X_k>v_k\}}1\!\!1_{\{S\leq u-c\}}}\right)} \\&=\mathbb{E}\left({g(u^*-c)1\!\!1_{\{S\leq u-c\}}}\right)+ \sum_{k=1}^{d}{\mathbb{E}\left({g(u^*_k-X_k)1\!\!1_{\{X_k>u^*_k\}}1\!\!1_{\{S\leq u-c\}}}\right)} \/,
 \end{align*}
 for all $k\in\{1,\ldots,d\}$ we put, for example, $\alpha_k=\alpha=\frac{u^*-c}{d}<0$, and since the function $g$ is convex and $g(0)=0$, it satisfies for all real $0<\beta<1$, $g(\beta x)\leq\beta g(x),\forall x\in\mathbb{R}^-$. Then: 
 \begin{align*}
  I(u^*,u^*_1,\ldots,u^*_d)&\geq\mathbb{E}\left({d\cdot g\left(\frac{u^*-c}{d}\right)1\!\!1_{\{S\leq u-c\}}}\right)+ \sum_{k=1}^{d}{\mathbb{E}\left({g(u^*_k-X_k)1\!\!1_{\{X_k>u^*_k\}}1\!\!1_{\{S\leq u-c\}}}\right)}\\ &=\mathbb{E}\left({d\cdot g\left(-(-\alpha)_+\right)1\!\!1_{\{S\leq u-c\}}}\right)+ \sum_{k=1}^{d}{\mathbb{E}\left({g(-(X_k-u^*_k)_+)1\!\!1_{\{S\leq u-c\}}}\right)}\\&= \sum_{k=1}^{d}{\mathbb{E}\left({[g(-(X_k-u^*_k)_+)+g(-(-\alpha_k)_+)]1\!\!1_{\{S\leq u-c\}}}\right)}\/,\end{align*}
  $x\rightarrow g(-(x)_+)$ is also a 1-homogeneous convex function, then:
  \begin{align*}
    I(u^*,u^*_1,\ldots,u^*_d)
  \geq  \sum_{k=1}^{d}{\mathbb{E}\left({g(-(X_k-(u^*_k+\alpha_k))_+)1\!\!1_{\{S\leq u-c\}}}\right)}\/,
  \end{align*}
  we remark that $\sum_{k=1}^{d}(u^*_k+\alpha_k)=u-c$, then $(u_1^*+\alpha,\ldots,u_d^*+\alpha)\in \mathcal{U}^{d}_{u-c}$.\\
  So, 
   \begin{align*}
    I(u^*,u^*_1,\ldots,u^*_d) &\geq  \sum_{k=1}^{d}{\mathbb{E}\left({g((u^*_k+\alpha_k)-X_k))1\!\!1_{\{X_k>u^*_k+\alpha_k\}}1\!\!1_{\{S\leq u-c\}}}\right)}\\
    &\geq \underset{v \in \mathcal{U}^{d}_{u-c}}{\inf}\sum_{k=1}^{d}{\mathbb{E}\left({g(v_k-X_k)1\!\!1_{\{X_k>v_k\}}1\!\!1_{\{S\leq u-c\}}}\right)}\\
    &=\sum_{k=1}^{d}{\mathbb{E}\left({g(u_k-X_k)^+1\!\!1_{\{S\leq u-c\}}}\right)}\\
    &= I(c,u_1,\ldots,u_d) \/,  
    \end{align*}
then,\[ I(u^*,u^*_1,\ldots,u^*_d)\geq I(c,u_1,\ldots,u_d) \/.\] That is contradictory with the uniqueness of the minimum on the set  $\mathcal{U}^{d+1}_u$.
 \item \textit{Case 2: $u^*>c$}\\
We have : 
\[ I(u^*,u^*_1,\ldots,u^*_d)= \sum_{k=1}^{d}{\mathbb{E}\left({g(u^*_k-X_k)1\!\!1_{\{X_k>u^*_k\}}1\!\!1_{\{S\leq u-c\}}}\right)} \/,\]
and, 
\[ I(c,u_1,\ldots,u_d)= \sum_{k=1}^{d}{\mathbb{E}\left({g(u_k-X_k)1\!\!1_{\{X_k>u_k\}}1\!\!1_{\{S\leq u-c\}}}\right)}\/. \]
Let $\alpha=\frac{u^*-c}{d}>0$, we remark that,$(u^*_1+\alpha,\ldots,u^*_d+\alpha)\in \mathcal{U}^{d}_{u-c}$, and that the penalty function $g$ is decreasing on $\mathbb{R}^-$ because $g^{\prime\prime}(x)\geq 0, \forall x\in \mathbb{R}^-$ and $g^\prime(0^+)= 0$.
Then,
  \begin{align*}
     I(c,u_1,\ldots,u_d)&= \sum_{k=1}^{d}{\mathbb{E}\left({g(-(X_k-u_k)_+)1\!\!1_{\{S\leq u-c\}}}\right)}\\&= \underset{v \in \mathcal{U}^{d}_{u-c}}{\inf}\sum_{k=1}^{d}{\mathbb{E}\left({g(v_k-X_k)1\!\!1_{\{X_k>v_k\}}1\!\!1_{\{S\leq u-c\}}}\right)}\\
         &\leq\sum_{k=1}^{d}{\mathbb{E}\left({g(-(X_k-(u_k^*+\alpha))_+)1\!\!1_{\{S\leq u-c\}}}\right)}\\
     &< \sum_{k=1}^{d}{\mathbb{E}\left({g(-(X_k-u^*_k))_+1\!\!1_{\{S\leq u-c\}}}\right)}\\
     &=\sum_{k=1}^{d}{\mathbb{E}\left({g(u^*_k-X_k)1\!\!1_{\{X_k>u^*_k\}}1\!\!1_{\{S\leq u-c\}}}\right)}= I(u^*,u^*_1,\ldots,u^*_d) \/.  
     \end{align*}
 That is contradictory with the fact that $I(u^*,u^*_1,\ldots,u^*_d)=\underset{v \in \mathcal{U}^{d+1}_u}{\inf}\mathit{I}(v)$.\\
 We deduce that the only possible case is the third one $u^*=c$.  
 \item \textit{Case 3: $u^*=c$}\\
 The uniqueness of the minimum implies that: \begin{align*}
                (u^*,u_1^*,\ldots,u_d^*)=(c,u_1^*,\ldots,u_d^*)&=\displaystyle{\underset{\mathcal{U}^{d+1}_{ u}}{\arg\min}}\sum_{k=1}^{d}\mathbb{E}[g(u_k-X_k)1\!\!1_{\{X_k> u_k\}}1\!\!1_{\{S\leq u-c\}}]\\
                &= \displaystyle{\underset{\mathcal{U}^{d}_{ u-c}}{\arg\min}}\sum_{k=1}^{d}\mathbb{E}[g(u_k-X_k) 1\!\!1_{\{X_k> u_k\}}1\!\!1_{\{S\leq u-c\}}]\\
                &= (c,u_1,\ldots,u_d)\/.
                \end{align*}
 Finally, we have proven that:
 \[ (u^*,u_1^*,\ldots,u_d^*)=(c,u_1,\ldots,u_d) \/.\]
 \end{itemize}
 \end{proof}
 Lemma \ref{sensSub} is related to the sub-additivity property. It will be used in the proof of the comonotonic additivity property. 
  \begin{lemma}
  \label{sensSub}
  Under Assumptions H1,H2 and H3, and for all $(i,j)\in \{1,\ldots,d\}^2$, and where, $x.e_i$ is the dot product of the vector $x\in\mathbb{R}^d$ and the $i^{th}$ component of the canonical basis of $\mathbb{R}^d$.
  \begin{itemize}
  \item if $ A_{X_1,\ldots,X_{i-1},X_i+X_j,X_{i+1},\ldots,X_{j-1},X_{j+1},\ldots,X_d}(u).e_i< A_{X_1,\ldots,X_d}(u).(e_i+e_j) $, then: \[ \forall k\in \{1,\ldots,d\}\setminus{i,j},~~ A_{X_1,\ldots,X_{i-1},X_i+X_j,X_{i+1},\ldots,X_{j-1},X_{j+1},\ldots,X_d}(u).e_k>A_{X_1,\ldots,X_d}(u).e_k\/,\]
  \item if $ A_{X_1,\ldots,X_{i-1},X_i+X_j,X_{i+1},\ldots,X_{j-1},X_{j+1},\ldots,X_d}(u).e_i> A_{X_1,\ldots,X_d}(u).(e_i+e_j) $, then: \[ \forall k\in \{1,\ldots,d\}\setminus{i,j},~~ A_{X_1,\ldots,X_{i-1},X_i+X_j,X_{i+1},\ldots,X_{j-1},X_{j+1},\ldots,X_d}(u).e_k<A_{X_1,\ldots,X_d}(u).e_k\/,\]
  \item if $ A_{X_1,\ldots,X_{i-1},X_i+X_j,X_{i+1},\ldots,X_{j-1},X_{j+1},\ldots,X_d}(u).e_i = A_{X_1,\ldots,X_d}(u).(e_i+e_j) $, then: \[ \forall k\in \{1,\ldots,d\}\setminus{i,j},~~ A_{X_1,\ldots,X_{i-1},X_i+X_j,X_{i+1},\ldots,X_{j-1},X_{j+1},\ldots,X_d}(u).e_k=A_{X_1,\ldots,X_d}(u).e_k\/.\]
  \end{itemize}     
  \end{lemma} 
  \begin{proof}
  In order to simplify the notation, and without loss of generality, we assume $i=d-1$ and $j=d$. 
  We put,$(u_1,\ldots,u_{d-1},u_d)=A_{X_1,\ldots,X_d}(u)$ and $(u_1^*,\ldots,u_{d-2}^*,u_{d-1}^*)=A_{X_1,\ldots,X_{d-2},X_{d-1}+X_{d}}(u)$.\\
   The optimality condition for $(u_1,\ldots,u_{d-1},u_d)$ is given $\forall (i,j)\in\{1,\ldots,d\}^2$ by equation \ref{Optcg}:
         \[ \mathbb{E}[g^\prime_i(u_i-X_i)1\!\!1_{\{X_i>u_i\}}1\!\!1_{\{S\leq u\}}]=\mathbb{E}[g^\prime_i(u_j-X_j)1\!\!1_{\{X_j>u_j\}}1\!\!1_{\{S\leq u\}}]=\lambda\/,\]
        and for  $(u_1^*,\ldots,u_{d-2}^*,u_{d-1}^*)$ is $\forall i\in\{1,\ldots,d-2\}$
        \[ \mathbb{E}[g^\prime_i(u^*_i-X_i)1\!\!1_{\{X_i>u^*_i\}}1\!\!1_{\{S\leq u\}}]=\mathbb{E}[g^\prime_i(u^*_{d-1}-(X_{d-1}+X_d))1\!\!1_{\{X_d+X_{d-1}>u^*_{d-1}\}}1\!\!1_{\{S\leq u\}}]=\lambda^*\/. \]
 Now, we suppose that $u^*_{d-1}>u_d+u_{d-1}$. In this case there exists $k\in\{1,\ldots,d-2\}$ such that $u^*_k<u_k$, and since the function $x\rightarrow g^\prime(-(x)_+)$ is decreasing on $\mathbb{R}^+$, then:
  \[ \mathbb{E}[g^\prime_i(u_k-X_k)1\!\!1_{\{X_k>u_k\}}1\!\!1_{\{S\leq u\}}]=\lambda< \mathbb{E}[g^\prime_i(u^*_k-X_k)1\!\!1_{\{X_k>u^*_k\}}1\!\!1_{\{S\leq u\}}]=\lambda^* \]
  we deduce from this that for all $k\in{1,\ldots,d-2}$ : $u_k^*<u_k$.\\ 
  The proof is the same if we suppose that $u^*_{d-1}<u_d+u_{d-1}$, and the additive case is a corollary of the two previous ones.
  \end{proof}
  \begin{Proposition}[Comonotonic additivity]
         Under Assumption H2, and for $g_k(x)=|x|$, for all $k\in \{1,\ldots,d\}$, if $r \leq d$ risks $X_{i_{i\in CR}}$ are comonotonic, then:
         \[ A_{X_{i_{i\in\{1,\ldots,d\}\setminus CR}},\sum_{k\in CR}^{}X_k}(u)=(u_{i_{i\in\{1,\ldots,d\}\setminus CR}},\sum_{k\in CR}^{}u_k)\/, \]
                         where $(u_1,\ldots,u_d)=A_ {X_1,\ldots, X_d} (u)$ is the optimal allocation of $u$ on the $d$ risks $(X_1,\ldots,X_d)$, $A_{X_{i_{i\in\{1,\ldots,d\}\setminus CR}},\sum_{k\in CR}^{}X_k}(u)$ is the optimal allocation of $u$ on the $n-d+1$ risks $(X_{i_{i\in\{1,\ldots,d\}\setminus CR}},\sum_{k\in CR}^{}X_k)$, and $CR$ denote the set of the $r$ comonotonic risk indexes.
         \end{Proposition}
       \begin{proof}
       For $(i,j)\in\{1,\ldots,d\}^2$, if $X_i$ and $X_j$ are comonotonic risks, then, there exists an increasing non negative function $h$ such that $X_i=h(X_j)$, and we remark that $h$ is strictly increasing under Assumption H2. Let $f$ be the function $x\rightarrow f(x)=x+h(x)$, so that $X_i+X_j=f(X_j)$. \\ We denote $(u_1,\ldots,u_d)=A_{X_1,\ldots,X_d}(u)$ and $(u_1^*,\ldots,u_{d-1}^*)=A_{X_{i\in\{1,\ldots,d\}\setminus\{i,j\},X_i+X_j}}(u)$, then,\\ $A_{X_{i\in\{1,\ldots,d\}\setminus\{i,j\},X_i+X_j}}(u).e_{d-1}=u^*_{d-1}$ and $A_{X_1,\ldots,X_d}(u).(e_i+e_j)=u_i+u_j$.\\
       From the optimality condition for the allocation $A_{X_1,\ldots,X_d}(u)$, given in Equation \ref{ZoneO}: \[ \mathbb{P}(X_i\geq u_i, S\leq u) = \mathbb{P}(X_j\geq u_j, S\leq u) \/,\] we deduce that $u_i=h(u_j)$ and that $u_i+u_j=f(u_j)$.\\
       If there exists $k\in\{1,\ldots,d\}\setminus\{i,j\}$, such that $u^*_k<u_k$, then $\forall k\in\{1,\ldots,d\}\setminus\{i,j\}$: \[ \mathbb{P}(X_k\geq u^*_k, S\leq u) > \mathbb{P}(X_k\geq u_k, S\leq u)\/, \]
       so, \begin{align*}
       \mathbb{P}(X_i+X_j\geq u^*_{d-1}, S\leq u)&= \mathbb{P}(X_j\geq f^{-1}(u^*_{d-1}), S\leq u)\\&=\mathbb{P}(X_k\geq u^*_k, S\leq u)\\
       &>\mathbb{P}(X_k\geq u_k, S\leq u)\\
       &=\mathbb{P}(X_j\geq u_j, S\leq u)\/,
       \end{align*}
       finally, we deduce that: $f^{-1}(u^*_{d-1}) < u_j$, then $u^*_{d-1}<f(u_j)=u_i+u_j$ and, \\ $\sum_{k\in\{1,\ldots,d\}\setminus\{i,j\}}u^*_k<\sum_{k\in\{1,\ldots,d\}\setminus\{i,j\}}u_k$ which is absurd.\\
       In the same way, the case $u_k<u^*_k$ for $k\in\{1,\ldots,d\}\setminus\{i,j\}$ leads to the contradiction.\\
       Using Lemma \ref{sensSub}, and under Assumption H3, we deduce the optimal allocation for the other risks $X_k$, $k\in\{1,\ldots,d\}\setminus \{i,j\}$.\\
              The additivity property for two comonotonic risks can be trivially generalized to several comonotonic risks.
       \end{proof}

  Concerning the sub-additivity property, we have not yet managed to build a demonstration for this property. However, simulations using the optimization algorithm presented in Cénac et al. (2012) \cite{AR2}, seem to confirm the sub-additivity of the allocation by minimizing the indicators $I$ and $J$, even for classic examples of non sub-additivity of the risk measure VaR.
\begin{remark}
The previous properties have been demonstrated for the optimal allocation by minimizing the risk indicator $I$, they can be demonstrated with the same arguments for the optimal allocation by minimization of the indicator $J$.
\end{remark} 
   \subsection{Other desirable properties}
    In this section, we show that the optimal allocation by minimization of the indicators $I$ and $J$ satisfies some desirable properties. We consider the allocation by minimizing the multivariate risk indicator $I$, the proofs are almost the same in the case of the indicator $J$.
  
   \begin{Proposition}[Positive homogeneity]
  Under Assumption H1, and for 1-homogeneous penalty functions $g_k$, $k\in\{1,\ldots,d\}$, for any $\alpha\in \mathbb{R}^+$:
   \[ A_{\alpha X_1,\ldots,\alpha X_d}(\alpha u)=\alpha A_{X_1,\ldots,X_d}(u) \/.\]
   \end{Proposition}
      \begin{proof}
           Since the penalty functions are convex and 1-homogeneous, then, for any $\alpha\in \mathbb{R}^{*+}$ :\begin{align*}
           A_{\alpha X_1,\ldots,\alpha X_d}(\alpha u)&= \displaystyle{\underset{(u^*_1,\ldots,u^*_d)\in\mathcal{U}^d_{\alpha u}}{\arg\min}}\sum_{k=1}^{d}\mathbb{E}[g_k(u_k^*-\alpha X_k) 1\!\!1_{\{\alpha X_k>u_k^*\}}1\!\!1_{\{\alpha S\leq \alpha u\}}]\\
           &=\displaystyle{\underset{(u^*_1,\ldots,u^*_d)\in\mathcal{U}^d_{\alpha u}}{\arg\min}}\sum_{k=1}^{d}{\alpha}\mathbb{E}[g_k\left(\frac{u_k^*}{\alpha}-X_k\right) 1\!\!1_{\{X_k>\frac{u_k^*}{\alpha}\}}1\!\!1_{\{S\leq u\}}]\\
           &=\displaystyle{\underset{(u^*_1,\ldots,u^*_d)\in\mathcal{U}^d_{\alpha u}}{\arg\min}}\sum_{k=1}^{d}\mathbb{E}[g_k\left(\frac{u_k^*}{\alpha}-X_k\right) 1\!\!1_{\{X_k>\frac{u_k^*}{\alpha}\}}1\!\!1_{\{S\leq u\}}]\\
           &=\alpha\displaystyle{\underset{(u_1,\ldots,u_d)\in\mathcal{U}^d_{u}}{\arg\min}}\sum_{k=1}^{d}\mathbb{E}[g_k(u_k-X_k) 1\!\!1_{\{X_k> u_k\}}1\!\!1_{\{S\leq u\}}]\\
           &=\alpha A_{X_1,\ldots,X_d}(u)\/.
           \end{align*} 
      \end{proof}
   \begin{Proposition}[Translation invariance]
 Under Assumptions H1, H2 and for all $(a_1,\ldots,a_d)\in \mathbb{R}^d$, such that the joint density $f(X_k,S)$ support contains $[0, u+\sum_{k=1}^{d}a_k]^2$, for all $k\in\{1,\ldots,d\}$:
    \[ A_{X_1-a_1,\ldots,X_d-a_d}(u)=A_{X_1,\ldots,X_d}\left(u+\sum_{k=1}^{d}a_k\right)-(a_1,\ldots,a_d) \/.\]
  \end{Proposition}
  \begin{proof} We denote by $(u^*_1,\ldots,u^*_d)$ the optimal allocation $A_{X_1-a_1,\ldots,X_d-a_d}(u)$, and by $(u_1,\ldots,u_d)$ the optimal allocation $A_{X_1,\ldots,X_d}\left(u+\sum_{k=1}^{d}a_k\right)$.\\  Using the optimality condition (\ref{Optcg}), $(u^*_1,\ldots,u^*_d)$ is the unique solution in $\mathcal{U}^d_u$ of the following equations system:
  \[ \mathbb{E}[g^\prime_i(u^*_i-(X_i-a_i))1\!\!1_{\{X_i-a_i>u^*_i\}}1\!\!1_{\{S-a\leq u\}}]=\mathbb{E}[g^\prime_j(u^*_j-(X_j-a_j))1\!\!1_{\{X_j-a_j>u^*_j\}}1\!\!1_{\{S-a\leq u\}}],~~\forall j\in\{1,\ldots,d\} 
       \/,\]
       where $a=\sum_{k=1}^{d}a_k$. Then, $(u^*_1,\ldots,u^*_d)$ satisfies also:
        \[ \mathbb{E}[g^\prime_i(u^*_i+a_i-X_i)1\!\!1_{\{X_i>u^*_i+a_i\}}1\!\!1_{\{S\leq u+a\}}]=\mathbb{E}[g^\prime_i(u^*_i+a_j-X_j)1\!\!1_{\{X_j>u^*_j+a_i\}}1\!\!1_{\{S\leq u+a\}}],~~\forall j\in\{1,\ldots,d\} 
              \/.\]
              Since, $(u^*_1+a_1,\ldots,u^*_d+a_d)\in\mathcal{U}^d_{u+a}$, and from the solution uniqueness of the optimality condition~(\ref{Optcg}) for the allocation $A_{X_1,\ldots,X_d}\left(u+a\right)$, we deduce that: $u^*_k+a_k=u_k$ for all $k\in\{1,\ldots,d\}$.
             \end{proof}
   \begin{Proposition}[Continuity]
       Under Assumptions H1 and H2, and if $\forall k\in\{1,\ldots,d\}$, $\exists \epsilon_0>0$ such that:
       \[ \forall \epsilon, |\epsilon|<\epsilon_0,~~\mathbb{E}[\underset{v\in[0,u]}{\sup}|g_k^\prime(v-(1+\epsilon)X_k)|]<+\infty \/,\]
        then, if $(X_1,\ldots,X_d)$ is a continuous random vector, for all $i \in \{1,\ldots,d\}$:
      \[ \lim\limits_{\epsilon \to 0}A_{X_1,\ldots,(1+\epsilon)X_i,\ldots,X_d}(u)=A_{X_1,\ldots,X_i,\ldots,X_d}(u) \/.\]
   \end{Proposition}
    \begin{proof}
    Let $(u_1,\ldots,u_d)$ be the optimal allocation of $u$ on the $d$ risks $(X_1,\ldots,X_d)$:
    \[  (u_1,\ldots,u_d)=A_ {X_1,\ldots, X_i,\ldots, X_d} (u) \/,\]
    then $(u_1,\ldots,u_d)$ is the unique solution in $\mathcal {U} ^ d_ {u}$ of Equations system (\ref{Optcg}):
    \begin{equation*}
    \mathbb{E}[g^\prime_i(u_i-X_i)1\!\!1_{\{X_i>u_i\}}1\!\!1_{\{S\leq u\}}]=\mathbb{E}[g^\prime_i(u_j-X_j)1\!\!1_{\{X_j>u_j\}}1\!\!1_{\{S\leq u\}}],~~\forall j\in\{1,\ldots,d\} \/.
    \end{equation*}
   For $\epsilon\in\mathbb{R}$, let $(u^\epsilon_1,\ldots,u^\epsilon_d)$ be the optimal allocation of $u$ on the $d$ risks $(X_1,\ldots,X_{i-1},(1+\epsilon)X_i,X_{i+1},\ldots,X_d)$:
   \[  (u^\epsilon_1,\ldots,u^\epsilon_d)=A_ {X_1,\ldots,X_{i-1},(1+\epsilon)X_i,X_{i+1},\ldots,X_d}(u) \/,\] 
 then $(u^\epsilon_1,\ldots,u^\epsilon_d)$ is the unique solution in $\mathcal{U}^d_ {u}$ of the following equations system:
     \[ \mathbb{E}[g^\prime_i(u^\epsilon_i-(1+\epsilon)X_i)1\!\!1_{\{(1+\epsilon)X_i>u^\epsilon_i\}}1\!\!1_{\{S+\epsilon X_i\leq u\}}]=\mathbb{E}[g^\prime_i(u^\epsilon_j-X_j)1\!\!1_{\{X_j>u^\epsilon_j\}}1\!\!1_{\{S+\epsilon X_i\leq u\}}],~~\forall j\in\{1,\ldots,d\} \/.\]
Since $\mathcal{U}^d_{u}$ is a compact on $(\mathbb{R}^+)^d$, we may consider a convergent subsequence $(u^{\epsilon_k}_1,\ldots,u^{\epsilon_k}_d)$ of $(u^\epsilon_1,\ldots,u^\epsilon_d)$.\\
Since the penalties functions satisfy:
       \[\exists \epsilon_0>0,~~ \forall \epsilon, |\epsilon|<\epsilon_0,~~\mathbb{E}[\underset{v\in[0,u]}{\sup}|g_k^\prime(v-(1+\epsilon)X_k)|]<+\infty\/, \]
we use Lebesgue's dominated convergence Theorem to get:
  \[ \mathbb{E}[g^\prime_i(\lim\limits_{\epsilon \to 0}u^{\epsilon_k}_i-X_i)1\!\!1_{\{X_i>\lim\limits_{\epsilon \to 0}u^{\epsilon_k}_i\}}1\!\!1_{\{S\leq u\}}]=\mathbb{E}[g^\prime_i(\lim\limits_{\epsilon \to 0}u^{\epsilon_k}_j-X_j)1\!\!1_{\{X_j>\lim\limits_{\epsilon \to 0}u^{\epsilon_k}_j\}}1\!\!1_{\{S\leq u\}}],~~\forall j\in\{1,\ldots,d\}\/, \]
  thereby $(\lim\limits_{\epsilon \to 0}u^{\epsilon_k}_1,\ldots,\lim\limits_{\epsilon \to 0}u^{\epsilon_k}_d)$ is a solution of Equation (\ref{Optcg}), because  $\sum_{l=1}^{d}\lim\limits_{\epsilon \to 0}u^{\epsilon_k}_l=\lim\limits_{\epsilon \to 0}\sum_{l=1}^{d}u^\epsilon_l=u$, $(\lim\limits_{\epsilon \to 0}u^{\epsilon_k}_1,\ldots,\lim\limits_{\epsilon \to 0}u^{\epsilon_k}_d)\in\mathcal{U}^d_{u}$.\\ From the solution uniqueness of (\ref{Optcg}) in $\mathcal{U}^d_{u}$, we deduce that: $\lim\limits_{k \to \infty}u^{\epsilon_k}_i=u_i$ for all $i\in\{1,\ldots,d\}$.\\
  For all convergent subsequence of $(u_1^\epsilon,\ldots,u_d^\epsilon)$ the limit point is $(u_1,\ldots,u_d)$, we deduce that:\[ \lim\limits_{\epsilon \to 0}(u_1^\epsilon,\ldots,u_d^\epsilon)=(u_1,\ldots,u_d) \/.\]      \end{proof}
\begin{Proposition}[Monotonicity]
   Under Assumption H2, and for $(i,j)\in \{1,\ldots,d\}^2$, such that $g_i=g_j=g$:
   \[ X_i\leq_{st}X_j \Rightarrow u_i\leq u_j \/.\]
 \end{Proposition}
   \begin{proof}
     Let $(u_1,\ldots,u_j)$ be the optimal allocation $A_{X_1,\ldots,X_i,\ldots,X_d}(u)$, under Assumption H2, the optimality condition (\ref{Optcg}) is written as follows:
     \[ \mathbb{E}[g^\prime(u_i-X_i)1\!\!1_{\{X_i>u_i\}}1\!\!1_{\{S\leq u\}}]=\mathbb{E}[g^\prime(u_j-X_j)1\!\!1_{\{X_j>u_j\}}1\!\!1_{\{S\leq u\}}] \/.\]
    Now if $X_i\leq_{st}X_j$, and since, $x\mapsto -g^\prime(-(u_i-x)_+)1\!\!1_{\{S\leq u\}}$ is an increasing function on $\mathbb{R}^+$, then:
   \[ \mathbb{E}[g^\prime(u_i-X_j)1\!\!1_{\{X_j>u_j\}}1\!\!1_{\{S\leq u\}}]\leq\mathbb{E}[g^\prime(u_i-X_i)1\!\!1_{\{X_i>u_i\}}1\!\!1_{\{S\leq u\}}] \/.\] 
     We deduce that:
     \[\mathbb{E}[g^\prime(u_i-X_j)1\!\!1_{\{X_j>u_j\}}1\!\!1_{\{S\leq u\}}]\leq \mathbb{E}[g^\prime(u_j-X_j)1\!\!1_{\{X_j>u_j\}}1\!\!1_{\{S\leq u\}}]\/,\]
     and since, $g^\prime$ is an increasing function, and the distributions are all continuous, that implies: $u_j\geq u_i$.
      \end{proof}
      
   By combining all the properties demonstrated in this section, we show that in the case of penalty functions $g_k(x)=|x|~ \forall k\in \{1,\ldots,d\}$, and for continuous random vector $(X_1,\ldots,X_d)$, such that the joint density $f_{(X_k,S)}$ support contains $[0,u]^2$, for at least one $k\in\{1,\ldots,d\}$, the optimal allocation by minimization of the indicators $I$ and $J$ is a symmetric riskless full allocation. It satisfies the properties of comonotonic additivity, positive homogeneity, translation invariance, monotonicity, and continuity.\\ These properties are therefore desirable from an economic point of view, the fact that they are satisfied by the proposed optimal allocation implies that this allocation method may well be used for the economic capital allocation between the different branches of a group, in terms of their actual participation in the overall risk, taking into account both their marginal distributions and their dependency structures with the remaining branches.
  \section{Discussion: What could be the best choice for a capital allocation?}
    
  In section 3, we tried to explain why the optimal allocation can be considered coherent from an economic point of view. Now, the most practical question is to define the best allocation method choice for an insurer.\\
  
  The first goal of Solvency 2 norms is the insurers' protection and the ORSA approach is based on the minimization of risk at both the local and global levels. The classical methods of risk allocation give the weight of each business line in the group risk. The optimal allocation is based on the global risk optimization. From this point of view, the optimal capital allocation seems more in coherence with the ORSA goal. Indeed, the optimal allocation gives a second risk management level, after the solvency capital requirement determination.\\ 
  
  The best allocation method choice depends finally on the risk aversion of the insurer. If the SCR is considered as the only risk management level, the classical methods of risk allocation are sufficient. If the insurer accepts to enhance his security level as it is the ORSA aim, the optimal allocation can be a good practical answer to this need.\\

  Conventional capital allocation methods are based on a chosen univariate risk measure, their properties derive from those of this risk measure. It seems more coherent in a multivariate framework to use directly a multivariate risk indicator, not only for risk measurement, but also for capital allocation.\\
  
  Another important criterion for allocation method choice is the capital nature. The allocation of an investment capital may be different from that of a solvency capital.\\
      
  \section*{Conclusion}
 

In this article, we have shown that the capital allocation method by minimization of multivariate risk indicators can be considered as coherent from an economic point of view. This method also illustrates the importance of the risky business portfolio choice and its impact on the management of the overall company capital.\\
 
  
 
 This method can be developed if one can construct some broader sets of multivariate risk indicators as this is the case for univariate risk measures.\\
 
 Finally, the choice of a capital allocation method remains a complex and crucial exercise because some methods may be better suited to deal with specific issues, others can lead to dangerously wrong financial decisions. In the case of the proposed optimal capital allocation, the risk management is at the heart of the allocation process, and the company can allocate its capital and reduces its overall risk at the same time. Its risk aversion is reflected by the choice of the multivariate risk indicator to minimize. That is why we think that from a risk management point of view, this method can be considered as more flexible.

\bibliographystyle{plain}
\bibliography{pap1}

   \end{document}